%% file: CLF_revision.tex
\documentclass[10pt]{article}
\input{def.tex}

\renewcommand{\bf}{\boldsymbol{f}}
\usepackage{orcidlink}

\commenttrue 
 \usepackage[pagewise]{lineno}

\def\endproof{\nobreak\kern5pt\nobreak\vrule height4pt width4pt
depth0pt \vskip4pt plus2pt}
\definecolor{ForestGreen}{RGB}{34,139,34}

\begin{document}
\title{A Gr\"obner Approach to Dual-Containing Cyclic Left Module $(\theta,\delta)$-Codes
  % $Rg/Rf\subset R/Rf$
  over Finite Commutative Frobenius Rings}
\author[1]{Hedongliang Liu {\orcidlink{0000-0001-7512-0654}} (\href{mailto:lia.liu@tum.de}{lia.liu@tum.de})}
\affil[1]{School of Computation Information and Technology, Technical University of Munich, D-80333 M\"unchen, Germany}
\author[2]{Cornelia Ott {\orcidlink{0009-0000-8740-1664}}(\href{mailto:cornelia.ott@dlr.de}{cornelia.ott@dlr.de})}
\affil[2]{Institute of Communications and Navigation, German Aerospace Center (DLR), D-82234 Oberpfaffenhofen-Wessling, Germany}
\author[3]{Felix Ulmer \orcidlink{0000-0002-4907-4479} (\href{mailto:felix.ulmer@univ-rennes1.fr}{felix.ulmer@univ-rennes1.fr})}
\affil[3*]{Univ Rennes, CNRS, IRMAR - UMR 6625, F-35000 Rennes, France}
\date{}
\maketitle
\blfootnote{The work of H.~Liu has been supported by the German Research Foundation (DFG) with a German Israeli Project Cooperation (DIP) under grants no.~PE2398/1-1, KR3517/9-1. The work of F.~Ulmer has been by French projects ANR-11-LABX-0020-01``Centre Henri Lebesgue''. %This work was done when C.~Ott was with the Institute of Communications Engineering, Ulm University, Germany.
}
\blfootnote{Keywords: self-dual code, dual-containing code, skew polynomial, module code, Gr\"obner basis, non-zero derivation, code over rings}

\begin{abstract}
  For a skew polynomial ring $R=A[X;\theta,\delta]$ where $A$ is a commutative Frobenius ring, $\theta$ an endomorphism of $A$ and $\delta$ a $\theta$-derivation of $A$, we consider cyclic left  module codes $\mathcal{C}=Rg/Rf\subset R/Rf$ where $g$ is a left and right divisor of   $f$ in $R$. In this paper, we derive a parity check matrix when $A$ is a finite commutative Frobenius ring using only the framework of skew polynomial rings. We consider rings $A=B[a_1,\ldots,a_s]$  which are free $B$-modules
  where the restriction of  $\delta$ and $\theta$ to $B$ are polynomial maps. If a Gr\"obner basis can be computed over $B$, then we show that all Euclidean and Hermitian dual-containing codes $\mathcal{C}=Rg/Rf\subset R/Rf$ can be computed using a Gr\"obner basis. We also give an algorithm to test if the dual code is again a  cyclic left  module code. We illustrate our approach for rings of order $4$ with non-trivial endomorphism and the Galois ring of characteristic $4$.
\end{abstract}

\section{Introduction}

For a (non-commutative) skew polynomial ring $R=A[X;\theta,\delta]$ where $A$ is a commutative Frobenius ring, $\theta$ an endomorphism of the ring $A$ and $\delta$ a $\theta$-derivation of $A$ (\cref{defder}), we consider codes that are cyclic left  modules (i.e., generated by one element in $R$) of the form $\cC=Rg/Rf\subset R/Rf$ with $f=hg$ for $g,f,h\in R$.
In order to obtain a  parity check matrix for such codes we will make the additional assumption that there exists $\hbar\in R$ such that $f=hg=g\hbar$ (i.e., $g$ is a left and right divisor of $f$).
A parity check matrix of such codes has been derived in a general approach
in  \cite{boulagouaz2013delta}, in \cite[Corollary 4]{BUmodule} for $A=\FF_q$ and in \cite{BD,BD2} and for  a finite commutative Frobenius ring $A$. The framework of pseudo-linear transformations is used in \cite{BD,BD2,boulagouaz2013delta}, while in this paper we only use the framework of skew polynomial rings.
The entries of the parity check matrix are expressions in images under compositions of $\theta$ and $\delta$ of the coefficients of $\hbar$ and~$g$, which are difficult to solve when searching for self-dual/dual-containing codes.

If $A=B[a_1,\ldots,a_s]$  is a free $B$-module and  the restriction of  $\delta$ and $\theta$ to $B$ are polynomial maps (\cref{polmaps}) then we can  transform algebraic expressions of images under $\theta$ and $\delta$ into polynomial expressions over $B$. By representing the unknown coefficients of $\hbar$ and~$g$ as a linear combination of the algebra basis $B[a_1,\ldots,a_s]$, we obtain multivariate polynomial expressions for the entries of a parity check matrix.
The smallest unitary subring of a finite commutative ring $A$ (the image of the canonical map $\ZZ\to A$  given by $1\mapsto 1$) is either isomorphic to a finite field $\FF_p$ of prime order or to an integer modular ring $\ZZ_m=\ZZ/(m)$ (note that $p$ or $m$ here is the characteristic of the ring $A$).
Since  $\theta$ and $\delta$ are polynomial maps over the smallest subalgebra $B\subset A$  and since polynomial equations over $B$ can be solved using a Gr\"obner basis \cite{adams1994introduction,magma}, our approach applies to many rings $A$.
Within the computation complexity constraints of the Gr\"obner basis, we can find all dual-containing codes $\cC=Rg/Rf\subset R/Rf$ over $A$ for any given parameters $[n,k]$. Using this approach we can also test various properties of codes $\cC=Rg/Rf\subset R/Rf$.
   \begin{enumerate}
   \item In \cref{algo2}, we give an algorithm to compute, within the capability of the  Gr\"obner basis computation,  all (Hermitian-) dual-containing codes   ${\cC=Rg/Rf\subset R/Rf}$ over $A$ for any given parameters $[n,k]$. Dual-containing codes play an important role in constructing quantum error-correcting codes.
   \item  In \cref{tab:HammingF2u2} and \cref{tab:HammingF2v2}, we give examples of Hamming weight distribution of (Hermitian-) dual-containing codes ${\cC=Rg/Rf\subset R/Rf}$ that could not be found without considering either non-zero derivations or non-trivial endomorphisms.
   \item In \cref{results}, we give an example of a generating polynomial $g$ of a $[6,4]$ dual-containing code where all $8$ polynomials $f=hg=g\hbar$ of degree $6$ are non-central.
   \item  We give a procedure to decide whether the (Hermitian-)  dual code $\cC^{\perp}$ of  a cyclic left  module $\cC=Rg/Rf\subset R/Rf$  is again a cyclic left module code $Rg^\bot/Rf\subset R/Rf$, i.e., generated by a single polynomial $g^\bot$.
   \item  In \cref{algo1}, we use the previous algorithm to show that  many  (Hermitian-) dual-containing codes   $\cC=Rg/Rf\subset R/Rf$ have a dual code which is not a cyclic left module code.
      \end{enumerate}
 We apply our method to all commutative rings of order $4$: $A=\FF_2[v]/(v^2+v)$, ${A=\FF_2[u]/(u^2)}$, and $A=\FF_4$, which have a non-trivial endomorphism. We also give some examples in characteristic $4$ for the Galois ring  $$A=\mbox{GR}(2,2)=\ZZ_4[X]/(X^2+X+1).$$

\section{Preliminaries}
 \subsection{Skew Polynomial Rings}
Let $A\neq \{0\}$ be a \textbf{unitary ring} (i.e., there exists $1\in A\setminus \{0\}$, such that $1\cdot a=a\cdot 1=a, \forall a\in A$). We only consider \textbf{unitary endomorphisms} $\theta$ with the property that $\theta(1)=1$.
The identity automorphism will be denoted by $\Id$.
\begin{definition} \label{defder} Let $A\neq \{0\}$ be a unitary ring and $\theta$ an endomorphism of $A$.
A~\textbf{$\theta$-derivation} is a map $\delta:A\to A$ such that, for all $a,b\in A$
$$ \delta(a+b)  =\delta(a)+\delta(b)\quad \mbox{and} \quad \delta(a\,  b) = \delta(a)\, b+\theta(a)\, \delta(b). $$
A $\theta$-derivation $\delta$ is an \textbf{inner $\theta$-derivation} if there exists $\beta\in A$ with the property that $\delta(x)= \beta x - \theta(x)\beta$ for all $x\in A$. For $\beta=0$ we obtain the zero derivation which we denote $0$.
  \end{definition}
  We also use the exponential notation $\theta(a)=a{^\theta}$ and $\delta(a)=a^{\delta}$ throughout the paper.
Skew polynomial rings have been introduced and studied by Ore in \cite{Ore}.
A ~\textbf{skew polynomial ring} $R$ is defined as a set of (left) polynomials
%with regards to an endomorphism $\theta$ of $A$ and a $\theta$-derivation $\delta$, i.e.,
  $R= A[X;\theta,\delta]=\left\{\sum_{i=0}^na_iX^i\mid a_i \in A, n\in \NN \right \}$  with coefficients in the ring $A$.
The addition in $R$ is the usual polynomial addition and the multiplication is defined using the rule $Xa=a^\theta X+a^\delta$ which is extended using associativity and distributivity.

If the leading coefficient of $g\in R$ is invertible, then $\deg(hg)=\deg(h)+\deg(g)$ for all $h\in R$.
In the case of a non-division ring $A$, if~$\theta$ is an automorphism and the leading coefficient of $g\in R$ is invertible, then for any $f\in R$, we can perform a right (resp.~left) division of $f$ by $g$ and we obtain a unique right (resp.~left) remainder of degree $<\deg(g)$. To see this, let $g=\sum_{i=0}^m g_iX^i$, $f=\sum_{i=0}^nf_iX^i$ and $m\leq n$.
% we note that, if $m\le n$
Then for the right (resp.~left) division, the degree of $f-(f_n\theta^{n-m}(g_m^{-1})X^{n-m})g$ (resp.~$f-g\cdot(\theta^{-m}(g_m^{-1}f_n)X^{n-m})$) is less than the degree of $f$.
To show that the right\footnote{The uniqueness of the left division can be shown in the similar manner.} division of $f\in R$ by a fixed $g\in R$ is unique, suppose that $f=hg+r=\tilde{h}g+\tilde{r}$. This implies that $(h-\tilde{h})g=\tilde{r}-r$. If $h-\tilde{h}=c_sX^s+\ldots$ is non-zero, then the leading monomial in  $(h-\tilde{h})g$ is  $c_s\theta^s(g_m)X^{s+m}$. A unitary automorphism maps $1$ to $1$ and invertible elements to invertible elements, so that $c_s\theta^s(g_m)$ is non-zero. Since the right side of $(h-\tilde{h})g=\tilde{r}-r$ is of degree $<m$ we obtain $h=\tilde{h}$ and $r=\tilde{r}$.
Note that a unique right division of any $f\in R$ by a polynomial $g$ with invertible leading coefficient exists even if $\theta$ is just an endomorphism.

\subsection{Cyclic Left Module $(\theta,\delta)$-Codes $Rg/Rf\subset R/Rf$}
\begin{definition}
  Let $A$ be a  finite ring, $\theta$ an endomorphism of $A$, $\delta$  a $\theta$-derivation of $A$ and $R=A[X;\theta,\delta]$ the corresponding skew polynomial ring. Let $f\in R$ be a \textbf{monic} skew polynomial of degree $n$ and $g\in R$ be a right divisor of $f$.
A \textbf{cyclic left module $(\theta,\delta)$-code} (in short $(\theta,\delta)$-code) is defined as ${\cC}=Rg/Rf\subset R/Rf$ in the polynomial representation, and as
$$C= \{\bc=(c_0, c_1,\dots, c_{n-1}) \ |\  c_0+c_1X+\cdots+c_{n-1}X^{n-1}\in Rg/Rf\}$$
in the vector representation.
\end{definition}
Note that  $(\theta,\delta)$-codes with endomorphism and  derivation have been studied in  \cite{BUdelta,BD,BD2,ShBh}.   For  $\delta=0$ and $f=X^n-1$ we obtain  $\theta$-cyclic codes characterized by
$$\left(a_0,a_1,\ldots,a_{n-2},a_{n-1}\right)\in C\, \Rightarrow \, \left(\theta(a_{n-1}),\theta(a_0),\theta(a_1),\ldots,\theta(a_{n-2})\right)\in C$$
while the classical cyclic code correspond to $f=X^n-1$, $\theta=\mbox{id}$ and $\delta=0$ (in this case $R$ is a commutative univariate ring).
\begin{proposition}
  \label{prop:module}
  Let $g\in R=A[X;\theta,\delta]$ of degree $n-k$ be a right divisor of a monic skew polynomial $f\in R$ of degree $n$. A cyclic left module $(\theta,\delta)$-code ${\cC}=Rg/Rf\subset R/Rf$ has the following properties:
\begin{itemize}
\item The left $R$-module $R/Rf$ is also a free $A$-module isomorphic to $A^n$ with $A$-basis $(1,X,\ldots,X^{n-1})$.
\item $Rg/Rf$ is a left $R$-submodule of $R/Rf$.
\item  $Rg/Rf$ is also a free left $A$-submodule of $ R/Rf\cong A^n$ of dimension $k=\deg(f)-\deg(g)$ (i.e.~a linear code of length $n$ and dimension $k$ over the alphabet $A$).
\item If $g$ is a left and right divisor of $f=hg=g\hbar$, then we can assume that all the three polynomials in $f=hg$ for a cyclic left module $(\theta,\delta)$-code $Rg/Rf\subset R/Rf$ are monic (i.e.~have leading coefficient $1$). The monic generator polynomial $g$ of a cyclic left module $(\theta,\delta)$-code $Rg/Rf\subset R/Rf$ is unique.
\item If $g$ is a left and right divisor of $f=hg=g\hbar$ and $\theta$ is an automorphism, then we can assume that all the four polynomials in $f=hg=g\hbar$ for a cyclic left module $(\theta,\delta)$-code $Rg/Rf\subset R/Rf$ are monic.
\end{itemize}
\end{proposition}
\begin{proof}
  Since $f$ is monic we can perform a right division of any element of $R$ by $f$ and produce a unique remainder of degree $<n$. Therefore  any element of the  $R$-module $R/Rf$ has a unique representation (the remainder of the division) of degree $<n$ in $R/Rf$. Viewed as an $A$-module,  $R/Rf $ is a free $A$-module with basis $1,X,\ldots,X^{n-1}$ (i.e.~every element in $R/Rf$ is a unique linear combination of those elements). %This shows that $R/Rf $ is isomorphic to $A^n$ as an $A$-module \lia{maybe a typo here? should it be $R$-module?}\co{yes, I think so. Here we show that it's isomorphic to $A^n$ as an $R$-module}.

  For $g\in R$ we have $Rf\subset Rg$ if and only if $g$ is a right factor of $f$ and in this case  $Rg/Rf$ is a cyclic left $R$-submodule of $R/Rf$  generated by $g+Rf$.  The leading coefficient $g_{n-k}$ of $g=\sum_{i=0}^{n-k} g_iX^i$ is a right divisor of  $1$ and is therefore  invertible. This shows that  $\deg(hg)=\deg(h)+\deg(g)$ for all $h\in R$ and implies that  the $R$-module $Rg/Rf$ is a free $A$-module of rank $k$ of $A^n$  with basis $g,Xg,\ldots,X^{k-1}g$.

  We now prove the fourth statement for $g= g_{n-k}X^{n-k}+\cdots + g_0$ and $h=h_kX^k+\cdots +h_0$. If $f=hg=g\hbar$, since $f$ is monic, the leading coefficient of $f=hg$ is $h_k\theta^{k}(g_{n-k})=1$, showing that $h_k$ and $\theta^{k}(g_{n-k})$ are invertible.
  We obtain
\begin{align*}
  f &= \underbrace{(h_kX^k+\cdots + h_0)}_{h}\cdot\underbrace{(g_{n-k}X^{n-k}+\cdots + g_0)}_{g} \\
   &= \underbrace{(h_kX^k+\cdots + h_0)\cdot g_{n-k}}_{\tilde{h}}\cdot\underbrace{g_{n-k}^{-1}\cdot(g_{n-k}X^{n-k}+\cdots + g_0)}_{\tilde{g}}
\end{align*}
In this representation $\tilde{g}$ is a monic polynomial. Since the endomorphism $\theta$ maps $1$ to $1$, the product rule of $A[X;\theta,\delta]$ shows that the leading coefficient of $\tilde{h}\tilde{g}$ is the leading coefficient of $\tilde{h}$. Because   $\tilde{h}\tilde{g}=f$ is monic, we obtain that $\tilde{h}$ is also a monic polynomial.
The polynomials $g$ and $\tilde{g}$ differ by an invertible element. Using the above equations in both directions we see  that any multiple of $g$ by a polynomial of degree $\le k-1$ is also a multiple of $\tilde{g}$  by a polynomial of degree $\le k-1$ and vice versa. Therefore $Rg/Rf = R\tilde{g}/Rf$, showing that $g$ and $\tilde{g}$ generate the same codes. Hence, without loss of generality, we can assume that $f$, $h$ and $g$ are monic. If $C$ is a cyclic left module $(\theta,\delta)$-code  with parameters $[n,k]$ and  monic generator polynomial $g$, then any codeword is of the form $c=m\cdot g$ with $\deg(g)=n-k$ and $\deg(m)<k$. In particular the only monic polynomial of degree $n-k$ in $\cC$ is $g=1\cdot g$.

We now show that $\hbar=\hbar_kX^k+\cdots + \hbar_0$ is monic if $\theta$ is an automorphism. According to the above we can assume that $f$ and $g$ are monic, so that the  leading coefficient $\theta^{n-k}(\hbar_k)$  in $g\hbar=f$ must be $1$. Since $\theta^{n-k}$ is an automorphism, we obtain that $\hbar_k=1$.
  \end{proof}
  \begin{example} \label{eg:ord4_u2_0}
    The Frobenius chain ring $A=\FF_2[u]/(u^2)$ is a free $\FF_2$-module
    $\FF_2[u]$ with $\FF_2$ basis $[1,u]$.  The only automorphism of $A$ is the identity $\theta_1:x\mapsto x$. There is a unique endomorphism defined by $\theta_2(u)= 0$  (note that $\theta_2(1)=1$) which is a polynomial map on $\FF_2$ and on $A$ itself $\theta_2:x\mapsto x^2$.  Any $\theta$-derivation $\delta$ is  determined by $\delta(u)$ (note that $\delta(1)=\delta(0)=0$). The list of  $\theta$-derivation is:
 \[\begin{array}{|c||l||l|l|l|} \hline
    &Automorphism &Endomorphism \\ \hline
 & \theta_1=\Id & \theta_2 :u\mapsto 0  \\ \hline \hline
\delta_1=0 &  \cellcolor{gray!30}{u\mapsto 0} &   \cellcolor{gray!30}{u\mapsto 0} \\ \hline
\delta_2 & u\mapsto 1 &    \\ \hline
\delta_3 & u\mapsto u & \cellcolor{gray!30}{ u\mapsto u}  \\ \hline
\delta_4 & u\mapsto u+1 &  \\ \hline \end{array}\]
We marked the inner $\theta$-derivation  by a gray cell.

In \cref{prop:module} we showed that all the three polynomials in $f=hg$ for a cyclic left module $(\theta,\delta)$-code $Rg/Rf\subset R/Rf$  can be chosen to be monic. If $\theta$ is a non-trivial endomorphism, i.e., not an automorphism, the $\hbar$ in the decomposition $f=g\hbar$ is not necessarily monic. To see this, consider the ring $R=A[X;\theta_2(u)=0,\delta_3(u)=u]$ and the code $Rg/Rf$ with $g=X^2 + uX + u + 1$ and $f=X^4 + (u + 1)X^3 + X + u + 1$.
It can be found that
$f=hg$ with $h=X^2 + (u + 1)X + 1$, and $f=g\hbar$ with $\hbar = (u + 1)X^2 + (u + 1)X + u + 1$ or $\hbar= (u + 1)X^2 + X + u + 1$.
\end{example}

  The  \textbf{encoding} of the information $(b_0,b_1,\ldots,b_{k-1})\in {A}^k$ in a cyclic left module $(\theta,\delta)$-code ${\cC}=Rg/Rf\subset R/Rf$ is given by the coefficients of
 $(\sum_{i=0}^{k-1} b_i X^i) g\in R$.
     A  \textbf{generator matrix} of the code is of the following form:
     \begin{align*}
       G=
       \begin{pmatrix}
         g_0 & g_1 & \cdots & g_{n-k} & 0& \cdots & \\
         g_0^\delta & g_1^\delta+g_0^\theta & g_2^\delta+g_1^\theta &\cdots &g_{n-k}^\theta &  0 &\cdots \\
         \vdots & \ddots & \ddots & \cdots &\ddots &\ddots\\
         g_0^{\delta^{k-1}} &&\cdots &&\cdots && g_{n-k}^{\theta^{k-1}}
       \end{pmatrix}
     \end{align*}
     The rows are given by the coefficients of $g, X\cdot g,\dots, X^{k-1}\cdot g$ and can be computed using the rule $X a =a^\theta X +a^\delta $ for $a\in A$. In particular, the code is completely determined by $g$, $\theta$ and $\delta$.
\begin{example} \label{exgen} In the notations of the above definition, consider a unitary polynomial $f=h g$ in $R=A[X;\theta,\delta]$  of degree $4$ with $g=g_1X+g_0$,  $h=\sum_{i=0}^3h_iX^i$ and $h_3g_1=1$. The code ${\cC}=Rg/Rf\subset R/Rf$ is a
$[4,3]_A$ code whose generating matrix is $$G=\left(\begin{array}{cccc} g_0 & g_1 & 0 & 0\\ g_0^\delta & g_1^\delta+g_0^\theta & g_1^\theta &0\\
 g_0^{\delta^2} & g_0^{\delta\theta}+g_0^{\theta\delta}+g_1^{\delta^2} &g_0^{\theta^2}+ g_1^{\delta\theta}+g_1^{\theta\delta}& g_1^{\theta^2}
\end{array}\right).$$
 \end{example}
 If $\theta$ is of the form  $a\mapsto a^{p^m}$  and $\delta$ is an inner $\theta$-derivation $a\mapsto \beta a -\theta(a) \beta$ (those are the only possibilities if $A$ is a finite field $\FF_q$), then the entries of the above matrix become polynomials in the coefficients of $g$ and allow sophisticated computations. This is the reason why almost all known examples of self dual  $(\theta,\delta)$-code consider $A$ to be a finite field.
\begin{definition}[Hermitian Inner Product and Hermitian Dual]
  Let $\sigma$ be an automorphism of $A$ whose order divides $2$. The \textbf{$\sigma$-Hermitian inner product} of $\bx,\by\in A^n$ is defined  as $ \langle \bx,\by \rangle_\sigma=\sum_{i=1}^n x_i \sigma(y_i)$.
  The \textbf{($\sigma$-Hermitian) dual code} of a code $C$ is defined as
  \begin{align*}    {C}^{\perp_{\sigma}}=\{\bv\ |\  \langle \bv,\bc \rangle_\sigma=0,\ \forall \bc\in C\}\ .  \end{align*}
  A code is \textbf{$\sigma$-{dual-containing}} if ${C}^{\perp_{\sigma}}\subset C $ and \emph{$\sigma$-self dual} if $C={C}^{\perp_{\sigma}}$.
\end{definition}
  If $\sigma=\Id$, the identity automorphism, we obtain the Euclidean inner product over $A^k$ and the Euclidean dual. In this case, we omit the $\sigma$ in the notation.
\begin{lemma}
\label{wood} (\cite{wood_2017})
Let $A$ be a commutative Frobenius ring and $C$ be a linear code over $A$. Then for the (Hermitian) dual code $C^{\perp_{\sigma}}$ we have
${|C|\cdot |C^{\perp_{\sigma}}|=|A|^n}$.
\end{lemma}

\section{Parity Check Matrix and Generating Matrix of the Euclidean and Hermitian Dual Code}
A parity  check matrix of a cyclic left module $(\theta,\delta)$-code ${\cC}=Rg/Rf\subset R/Rf$ over a field $\FF_q$ (where all derivations must be inner) is given in \cite[Corollary 4]{BUmodule} and over a commutative Frobenius ring in  \cite{BD,BD2}.
In \cite{ShBh} a parity check matrix for  a cyclic left module $(\theta,\delta)$-code ${\cC}=Rg/Rf\subset R/Rf$ over the ring $A=(\ZZ/4\ZZ)[X]/(X^2-1)$ is studied when $f=hg$ is a central polynomial.
Our approach is similar to \cite{BUmodule} and \cite{ShBh}. In this case $Rg/Rf$ is an ideal in $R/Rf$. In the next theorem we follow the assumption in
\cite{BUmodule,boucher2011note,BD}
that $g$ is both a right and a left divisor of $f$, i.e., $f=hg=g\hbar$, in which case $Rg/Rf$ is usually only a submodule of  $R/Rf$. The assumption that $f=hg=g\hbar$ is much weaker than the assumption that $f$ is central (see the $[6,4]$ example in \cref{results}).
\begin{lemma}  \label{lem:Mmat}  Let  $\theta$ be an {endomorphism} of the  finite ring $A$, $\delta$ a $\theta$-derivation on $A$,  $R=A[X;\theta,\delta]$, $f\in R$ a monic polynomial  having  both a right and a left divisor $g$ (i.e.~$f=hg=g\hbar$) and $\cC=Rg/Rf\subset R/Rf$ a cyclic left module $(\theta,\delta)$-code.
A word in $A^n$ corresponding to an element $w\in R$ of degree $<n$ is a codeword of $\cC$ if and only if (the coset of) $w\cdot \hbar=0$ in $R/Rf$.
We obtain an $n\times n$ matrix $M$ such that the vector representation of $\cC$ is $C=\{\bw\in A^n \,| \,\bw M=\vec 0\}$ (\ie $C=\lker(M)$ is a left kernel of $M$), where the entries of $M$ are images under compositions of $\theta$ and $\delta$ of the coefficients of $\hbar$ and $g$.
The $i$-th row of $M$ corresponds to the coefficients of $X^{i-1}\hbar  \mod f$, for $i=1,\dots,n$.
\end{lemma}
\begin{proof}
Since $f$ is monic, an element of $A^n\cong R/Rf$ has a unique representation as a remainder of the right division by $f$, which corresponds to a polynomial $w\in R$ of degree $<n$. Following \cite[Lemma 8]{BU} we now show  that $w$ corresponds to a codeword if and only if (the coset of) $w\cdot \hbar=0$ in $R/Rf$. If $w\in {\cC}$, then $w=\tilde{w}g$ for some $\tilde{w}\in R$. Therefore  $w\hbar =\tilde{w}g\hbar=\tilde{w}(h\cdot g)$ showing that  $w\hbar =0$ in $R/Rf=R/R(h\cdot g)$. Conversely, if $w\hbar =0$ in $R/Rf=R/R(h\cdot g)$, then $w\hbar= \tilde{w} (h\cdot g)= \tilde{w} (g \hbar)$ for some $\tilde{w}\in R$. Since  $\hbar$ is not a zero divisor in $R$ we obtain $w = \tilde{w} g$, showing that $w\in {\cC}$.

Since $f$ is monic, a word in $A^n$ corresponds to a coset  $w=\sum_{i=0}^{n-1} a_iX^i\in R/Rf$. Such a coset $w=\sum_{i=0}^{n-1} c_i$ belongs to $\cC$ if and only if $w\hbar =0$ in $R/Rf$. The coefficients of
\begin{align*}
w\hbar\mod f=\left(\sum_{i=0}^{n-1} c_iX^i\right)\left(\sum_{i=0}^{k} \hbar_iX^i\right)\mod f
\end{align*}
are obtained by bringing the coefficients $\hbar_i$ to the left side and performing the right division of $f$.
The code $C$ corresponds to the left kernel of this linear system and can therefore be represented as the left kernel of a matrix $M$, i.e., $
  \sum_{i=0}^{n-1} c_iM_{i,j}X^j =0, \ \forall j\in[0,n-1]$.
The entries $M_{ij}\in A$ are the images under $\theta$ and $\delta$ of the coefficients of $\hbar$ and $f$.
Since $f=g\hbar$, the  entries $M_{ij}\in A$ are images under $\theta$ and $\delta$ of the coefficients of $\hbar$ and  $g$.
The $i$-the row of $M$ corresponds to the contribution of $c_iX^{i}\cdot\hbar  \mod f$ in the product $w\hbar =0 \mod f$.
\end{proof}
\begin{example}[A toy example for $n=3,k=1$] \label{example_parity}  Consider  a ring $A$, a skew polynomial ring $R=A[X;\theta,\delta]$ and $f=X^3+\sum_{i=0}^2f_i X^i\in R$ such that $f=g\hbar$ in $A[X;\theta,\delta]$ for  $g=g_2X^2+g_1X+g_0$ and $\hbar=\hbar_1X+\hbar_0$.
According to \cref{lem:Mmat}, $w=c_0+c_1X+c_2X^2$ belongs to  ${\cC}=Rg/Rf\subset R/Rf$ if and only if  $w\hbar=0$ in $R/Rf$, i.e.\ $w\hbar \equiv 0 \mod f$. Since
  \begin{align*}
    w\hbar\mod f= & \left(c_2 (\hbar_1^{\theta\delta}+\hbar_1^{\delta\theta}+\hbar_0^{\theta^2}{-\hbar_1^{\theta^2}f_2})+c_1\hbar_1^\theta\right)X^2\\
           & + \left( c_2(\hbar_1^{\delta^2}+\hbar_0^{\theta\delta}+\hbar_0^{\delta\theta}{-\hbar_1^{\theta^2}f_1}) + c_1(\hbar_1^\delta+\hbar_0^\theta) + c_0\hbar_1 \right)X\\
                     &  + c_2(\hbar_0^{\delta^2}{-\hbar_1^{\theta^2}f_0}) + c_1\hbar_0^\delta+c_0\hbar_0
  \end{align*}
we obtain the condition  $w\in  {\cC} \, \Leftrightarrow  \, \bw\cdot M=\0$ where $\bw = (c_0, c_1, c_2)$ and
\begin{eqnarray}
  M & = & \label{eq:exMn3k2}
   \left( \begin{array}{ccc}
\cellcolor{gray!30}{  \hbar_0 }& \cellcolor{gray!30}{\hbar_1} & \cellcolor{gray!30}{0 }\\
 \cellcolor{gray!30}{  \hbar_0^{\delta}}&
\cellcolor{gray!30}{  \hbar_1^\delta+\hbar_0^\theta}  & \cellcolor{gray!30}{\hbar_1^\theta}\\
\hbar_0^{\delta^2}{-\hbar_1^{\theta^2}f_0} &
\hbar_1^{\delta^2}+\hbar_0^{\theta\delta}+\hbar_0^{\delta\theta}{-\hbar_1^{\theta^2}f_1} &
  \hbar_1^{\theta\delta}+\hbar_1^{\delta\theta}+\hbar_0^{\theta^2}{-\hbar_1^{\theta^2}f_2}
    \end{array}\right).
\end{eqnarray}
Note that the entry $M_{ij}$ corresponds to the coefficient of the term $c_iX^j$ in the polynomial $w\hbar \mod f$.
\end{example}
The following result is contained in \cite{BD,BD2} where the result is proven using pseudo-linear transformation or Matrix-Product Codes, while we give a  proof within the setting of skew polynomial rings.
\begin{theorem} (cf.  \cite{BD,ShBh})
  \label{control}
  Let  $\theta$ be an {endomorphism} of the finite {Frobenius commutative}  ring $A$, $\delta$ a $\theta$-derivation on $A$,  $R=A[X;\theta,\delta]$ a skew polynomial ring, $f\in R$ a monic polynomial having a right and left divisor $g$ (i.e.~$f=hg=g\hbar$) and ${\cC}=Rg/Rf\subset R/Rf$ a cyclic left module $(\theta,\delta)$-code. Let $C$ be the vector representation of $\cC$. The dual code $C^\bot$ is a free $A$-module code and ${|C|\cdot |C^\bot|=|A|^n}$. There exists a parity check matrix $H$ for the code $C$ such that it is a generator matrix of the dual code $C^\bot$. The entries of the matrix $H$ are images under compositions of $\theta$ and $\delta$ of the coefficients of $\hbar$ and $g$.
\end{theorem}
\begin{proof}
    We denote by  $\tilde{C}$ the code generated by the columns of the  $n\times n$ matrix $M$ in \cref{lem:Mmat} whose left kernel is exactly the code $C$. Then the columns of $M$ form a generating set of $\tilde{C}$.
    By construction we have that $\tilde{C}\subset C^\bot$. Note that $\bw\in C$ if and only if $\bw M=\0$, which is equivalent to $\bw$  is orthogonal to all generators of $\tilde{C}$. Therefore $C=\tilde{C}^\bot$. By Lemma \ref{wood} we have that $|\tilde{C}^\bot| \cdot |\tilde{C}|=|A|^n$. From $C=\tilde{C}^\bot$ and $|C|=|A|^k$ we then get $|\tilde{C}|=|A|^{n-k}$. By Lemma \ref{wood} we also have that $|C^\bot | \cdot  |C|=|A|^n$, which implies that $|C^\bot| =|A|^{n-k}$. Since  $\tilde{C}\subset C^\bot$ and both codes have the same number of elements we get   $\tilde{C}= C^\bot$.

It is shown in \cref{lem:Mmat} that for $i=1,\dots, n-k$, the $i$-th row of $M$ corresponds to $X^{i-1} \hbar$ (e.g.~the gray part in \eqref{eq:exMn3k2}).
This shows that the right-upper $(n-k)\times (n-k)$ submatrix of $M$ is lower triangular with invertible diagonal elements $\hbar_k,\theta(\hbar_k),\ldots,\theta^{n-k-1}(\hbar_k)$ and the right-most $n-k$ columns of $M$ are therefore linearly independent. Hence, the $A$-submodule generated by the right-most $n-k$ columns of $M$ contains $|A|^{n-k}$ elements.
This shows that $C^\bot$ is a free $A$-module generated by the right-most $n-k$ columns of $M$ which (after transposing) form a parity check matrix of $C$ (see \eqref{eq:Hn3k2} for an example).
  \end{proof}
\begin{example}
According to \cref{control}, the dual of the code in \cref{example_parity} is generated by the right-most two columns $M$ in \eqref{eq:exMn3k2} which (after transposing) form a parity check matrix of $C$,
\begin{eqnarray}  \label{eq:Hn3k2}
  H & = &
   \left( \begin{array}{ccc}
{\hbar_1} &
{  \hbar_1^\delta+\hbar_0^\theta} &
{\hbar_1^{\delta^2}+\hbar_0^{\theta\delta}+\hbar_0^{\delta\theta}{-\hbar_1^{\theta^2}f_1}}\\
{0 }& {\hbar_1^\theta} &
{  \hbar_1^{\theta\delta}+\hbar_1^{\delta\theta}+\hbar_0^{\theta^2}{-\hbar_1^{\theta^2}f_2}}
    \end{array}\right).
\end{eqnarray}
 \end{example}
\begin{theorem} \label{th2}
  Let $\sigma$ be an automorphism of order $2$ of $A$, $\theta$ be an {endomorphism} of the finite {Frobenius commutative}  ring $A$, $\delta$ be a $\theta$-derivation on $A$,  $R=A[X;\theta,\delta]$ a skew polynomial ring, $f\in R$ a monic polynomial having a right and left divisor $g$ (i.e.~$f=hg=g\hbar$) and ${\cC}=Rg/Rf\subset R/Rf$ a cyclic left module $(\theta,\delta)$-code. Let $C$ be the vector representation of $\cC$ and $C^\bot$ the dual code of $C$.
  % is a free $A$-module code and ${|C|\cdot |C^\bot|=|A|^n}$.
  If we apply $\sigma$ to all entries of the generator matrix $G^\bot\coloneq H$ of the dual code $C^\bot$, then we obtain a generator matrix $G^{\bot_\sigma}$ of the $\sigma$-Hermitian dual code  $C^{\bot_\sigma}$ of $C$.  The coefficients of the matrix $\sigma(H)$ are expressions  in  images under compositions of $\theta$, $\delta$ {and $\sigma$} of the coefficients
of $\hbar$ and $g$.   \end{theorem}
\begin{proof}  For each row $g_s, s\in\{1,\dots,k\}$ of a generating matrix $G$ of $C$ and each row $g^\bot_t, t\in\{1,\dots,n-k\}$ of a generating matrix of $G^\bot$ of $C^\bot$ we have
$  \langle g_s,g^\bot_t \rangle=\sum_{i=1}^n g_{s,i}\,g^\bot_{t,i}=0$. Therefore
$$  \langle g_s,\sigma(g^\bot_t) \rangle_\sigma=\sum_{i=1}^n g_{s,i}\, \sigma(\sigma(g^\bot_{t,i}))=\sum_{i=1}^n g_{s,i}\, g^\bot_{t,i}=0.$$
Since the $n-k$ rows of $G^\bot$ generate a free code of dimension $|A|^{n-k}$ and $\sigma$ is an automorphism, the $n-k$ rows of $\sigma(G^\bot)$ also generate a free code of dimension $|A|^{n-k}$. Lemma \ref{wood} implies that they generate $C^{\bot_\sigma}$.
\end{proof}

\section{Computing all  Dual-Containing  $(\theta,\delta)$-Codes}
For the case $\delta=0$ and $A=\FF_q$, the generators of the dual code have been derived in \cite{BSU,BU,boucher2011note,BD,BD2}.
For an arbitrary Frobenius ring $A$ or $\delta\neq 0$, the dual code is much less studied. The algorithms in this section will allow us to show that the dual of a cyclic left module $(\theta,\delta)$-code is in general not a cyclic left module $(\theta,\delta)$-code.
\subsection{Polynomial Maps}
Our first goal is to transform algebraic expressions in the  images under $\theta$ and $\delta$ of the coefficients of $\hbar$ and $g$ (cf.\ $M$ in \eqref{eq:exMn3k2} above), into
multivariate polynomials over some subalgebra of $A$.
\begin{definition}
  \label{polmaps}
  A \textbf{polynomial map} on a ring $B$ is a map $f:B\to B;x\mapsto  \sum_{i=0}^s b_ix^i$, where $s\in \NN$ and $b_i\in B$.
\end{definition}
\begin{lemma}\label{lem:polymap} Let  $\theta$ be an endomorphism of the  finite   ring $A$, $\delta$ a $\theta$-derivation of $A$ and $B\subset A$ a subring. Let ${\cE}$ be a  system of finitely many equations over $A$ that are polynomial expressions in the  images under $\theta$ and $\delta$ of a finite set of variables $y_1,\ldots,y_m$.
  If  $A=B[a_1,\ldots,a_s]$ ($s\in \NN$) is a free $B$-module
  % free $B$-algebra
  and the restriction of  $\delta$ and $\theta$ to $B$ are polynomial maps, then  all solutions  in  $A^m$ of the system  ${\cE}$ correspond to the   solutions   in $B^{ms}$ of a
system of polynomial equations over $B$ in  the variables $y_{1,1},\ldots ,y_{1,s},\ldots,y_{m,1},\ldots,y_{m,s}$ where $y_i=y_{i,1}a_1+\cdots +y_{i,s}a_s$.
\end{lemma}
\begin{proof}
The image of the $B$-basis $a_1,\ldots,a_s$ of $A$ under  $\delta$ and $\theta$ are expressions  of the form
$\theta(a_i)=\gamma_{i,1}a_1+\cdots +\gamma_{i,s}a_s$ and $\delta(a_i)=\beta_{i,1}a_1+\cdots +\beta_{i,s}a_s$  for some given  $\gamma_{i,j}$ and $\beta_{i,j}$ in $B$:
\begin{align*}
  y_i^{\theta} & = (y_{i,1}a_1+\cdots +y_{i,s}a_s)^\theta  \\
               &=  y_{i,1}^{\theta}a_1^{\theta} + \cdots  +y_{i,s}^{\theta}a_s^{\theta}\\
               & = y_{i,1}^{\theta}\left(\gamma_{1,1}a_1+\cdots +\gamma_{1,s}a_s\right)+\cdots +y_{i,s}^{\theta}\left(\gamma_{s,1}a_1+\cdots +\gamma_{s,s}a_s\right)\\
  y_i^{\delta} &= (y_{i,1}a_1+\cdots +y_{i,s}a_s)^\delta \;=\; (y_{i,1}a_1)^\delta+\cdots +(y_{i,s}a_s)^\delta\\
               & = y_{i,1}^\delta a_1+ y_{i,1}^\theta a_1^\delta +\cdots +y_{i,s}^\delta a_s+y_{i,s}^\theta a_s^\delta\\
& = y_{i,1}^\delta a_1+  y_{i,1}^{\theta}\left(\beta_{1,1}a_1+\cdots +\beta_{1,s}a_s\right)+ \cdots +y_{i,s}^\delta a_s+y_{i,s}^{\theta}\left(\beta_{s,1}a_1+\cdots +\beta_{s,s}a_s\right)
\end{align*}
Using
\begin{enumerate}
\item the algebra relations $a_ia_j=\mu_{i,j,1}a_1+\ldots+\mu_{i,j,s}a_s$ (where $\mu_{i,j,s}\in B$ are given),
\item the additive and multiplicative properties of $\theta$ and $\delta$,
\item the fact that  the restriction of  $\delta$ and $\theta$ to $B$ are polynomial maps on $B$ (so that $y_{i,j}^\theta$ and $y_{i,j}^\delta$ are polynomials in $y_{i,j}$ over $B$),
\end{enumerate}
we can recursively transform any system of polynomial equations in the variables $y_1,\ldots,y_m$, whose solutions are in $A^m$, into a system of polynomial equations in the variables $y_{1,1},\ldots ,y_{1,s},\ldots y_{m,1}, \ldots,\\y_{m,s}$, whose solutions are in $B^{ms}$.
\end{proof}

The following two examples show that there are endomorphism $\theta$ and {$\theta$-derivation} $\delta$ of a ring $A$ which are not polynomial maps over $A$, but only polynomial maps over a subring $B$.

\begin{example} \label{ex_F2v}
  % Consider the Frobenius ring $A=\FF_2[v]/(v^2+v)$ of order $4$ whose automorphism group is of order two generated by $\theta_2$ which is defined by $\theta_2(v)=v+1$ (note that we have $\theta_2(1)=1$).
  % \lia{The first sentence seems to have some grammar mistake and not very concise. I rephrased as following:}
  Consider the Frobenius ring $A=\FF_2[v]/(v^2+v)$ of order $4$. There are two automorphisms $\theta_1=\Id$ and $\theta_2$ of order two, and two non-trivial endomorphisms $\theta_3$ and $\theta_4$. Any $\theta$-derivation $\delta$ is determined by $\delta(v)$ (note that $\delta(1)=\delta(0)=0$).  All the $\theta$-derivations are listed below:
   \[\begin{array}{|c||l|l||l|l|} \hline
   & \multicolumn{2}{c|}{Automorphism}  & \multicolumn{2}{c|}{Endomorphism} \\ \hline
 & \theta_1=\Id & \theta_2(v)=v+1 & \theta_3(v)= 0  & \theta_4(v)=1  \\ \hline \hline
\delta_1=0&  \cellcolor{gray!30}{v\mapsto 0}
& \cellcolor{gray!30}{v\mapsto 0}
&  \cellcolor{gray!30}{v\mapsto 0}\
& \cellcolor{gray!30}{v\mapsto 0}
\\ \hline
\delta_2 & & \cellcolor{gray!30}{v\mapsto 1}
& & \\ \hline
\delta_3 & & \cellcolor{gray!30}{v\mapsto v}
& \cellcolor{gray!30}{v\mapsto v}
& \\ \hline
\delta_4 & & \cellcolor{gray!30}{v\mapsto v+1}
& & \cellcolor{gray!30}{v\mapsto v+1}
 \\ \hline \end{array}\]
 The inner $\theta$-derivations are marked by a gray cell. Here all $\theta$-derivations are inner.

  Suppose that  the automorphism $\theta_2$ is a polynomial map on $A$ of the form
$$f: x\mapsto \sum_{i\in\NN_0} (\alpha_{i,1} v+ \alpha_{i,0}) x^i\, = \,  \sum_{i\in\NN_0} \alpha_{i,1} v x^i+  \sum_{i\in\NN_0} \alpha_{i,0} x^i \qquad (\alpha_{i,j}\in\FF_2).
$$ Then $\theta_2(0)=0 \Rightarrow  \alpha_{0,0}=0$. Since $\alpha_{i,j}\in \{0,1\}$, $f(v)$ is a multiple of positive powers of $v$. Since $v^2=v$ we get that $f(v)$ is a sum of $v$, which is either $v$ or $0$ in this ring. Since $\theta_2(v)=v+1$, we obtain that $\theta_2$ is not a polynomial map on $A$.
\end{example}
\begin{example}  \label{eg:ord4_u2} We keep the notation of \cref{eg:ord4_u2_0} for the ring  $A=\FF_2[u]/(u^2)$. The only automorphism of $A$ is the identity $\theta_1:x\mapsto x$ which is a polynomial map.
Suppose that a derivation $\delta$ of $A$ is given by a polynomial map $\delta:x\mapsto\sum_{i=0}^t b_ix^i$ over $A$ (where $b_i\in A$). Since $\delta(1)=0$, we must have $b_0=0$ in the polynomial map.
From $u^2=0$, we obtain $\delta:u\mapsto \sum_{i=1}^t b_iu^i = b_1u$.
Write $b_1=\beta_{1,1}u+\beta_{1,0}\in A$ for some $\beta_{1,1},\beta_{1,0}\in \FF_2$, then $\delta(u)= \beta_{1,1}u^2+\beta_{1,0} u=\beta_{1,0}u$, which can never be $u+1$ or $1$. Hence, $\delta_2(u)=1$ and $\delta_4(u)=u+1$ are not polynomial maps on $A$. For this ring we will always work over $B=\FF_2$ even for $\delta_1$ and $\delta_3$.
Codes of small length over $A$ are classified in \cite{DS,HU}.
\end{example}
\begin{lemma}
  \label{primering}
Automorphisms and derivations of a finite Frobenius ring $A$ are polynomial maps over the smallest unitary subring $B$ of $A$.
  % For a finite Frobenius ring $A$ any automorphism and any derivation is a polynomial map over the smallest unitary subring ${B}$.
\end{lemma}
\begin{proof}  The smallest unitary subring ${B}$ of $A$  is the image of the canonical map $\ZZ\to A$  given by $1\mapsto 1$ and  is either isomorphic to a finite field $\FF_p$ of prime order or to an integer modular ring $\ZZ_m=\ZZ/(m)$ (here $p$ or $m$ is the characteristic of the ring $A$). Since any automorphism $\theta$ is given by $x\mapsto x$ and $\theta$-derivation $\delta$ is given by $x\mapsto 0$ on  ${B}$, they are both polynomial maps on   ${B}$.
\end{proof}
\subsection{Computations via Gr\"obner Basis over $B\subset A$}
In this section we  assume that  $\theta$ and $\delta$ are polynomial maps over a subalgebra $B\subset A$ (this is always the case for the smallest unitary subring ${B}$ of $A$ by \cref{primering}) and that $A=B[a_1,\ldots,a_s]$ is a free $B$-module.
% free $B$-algebra.
This will enable us to transform any expression in $\theta$ and $\delta$ over $A$ into a polynomial expression over $B$ (\cref{lem:polymap}). The classical algorithm to solve systems of polynomial equations in a multivariate polynomial commutative ring $B[y_{1,1},\ldots ,y_{1,s},\ldots,y_{m,1},\ldots,y_{m,s}]$ is via \textbf{Gr\"obner basis}.  This algorithm exists in particular if $B$ is a field or an integer quotient ring (\cite{adams1994introduction,magma}), and therefore always over the smallest unitary subring ${B}$ of the finite unitary ring $A$ (\cref{primering}).

\subsubsection{An Algorithm to Compute All Dual-Containing  $(\theta,\delta)$-Codes}
\label{algo2}
We first express the unknown coefficients in $A$ of $g$ and $\hbar$ as linear combinations in a given $B$-basis of $A$ over $B$ with unknown coefficients $x_i$ in $B$. The expressions in images under compositions of $\theta$ and $\delta$ of the coefficients of $\hbar$ and~$g$ then become polynomials in the variables $x_i$ in $B$. We then obtain a parity check matrix whose coefficients are  polynomials in the variables $x_i$.  We can impose that $g$ divides $g\hbar$ on the right by imposing that all the coefficients of the remainder, whose entries are polynomials in the unknown $x_i$, to be zero. We can also impose  $C^\bot \subset C$ by imposing all the entries $\bM^\top\cdot \bM$, which are also polynomials in the unknown $x_i$, to be zero. All these conditions lead to a multivariate polynomial system in the unknowns $x_i$ with coefficients in $B$. If a Gr\"obner basis algorithm exists for the ring $B$, then we can compute all dual-containing cyclic left module $(\theta,\delta)$-codes ${\cC}=Rg/Rf\subset R/Rf$ for the fixed parameters $[n,k]$. Note that ${\cC}^\bot \subset{\cC}$ is a property of the code ${\cC}$ and is therefore independent of the choice of $h,f,\hbar$. In other words, if ${\cC}^\bot \subset{\cC}$ holds for some valid solution of $h,f,\hbar$, then this will hold for any valid solution of $h,f,\hbar$. We therefore simply have to compute an elimination basis for the possible polynomials $g$ and keep those skew polynomials $g$ that can be extended to a solution of the whole system.
\begin{algorithm}[htb!]
  \caption{Computing all {\em dual-containing} cyclic module $(\theta,\delta)$-codes for given $[n,k]$.}
  \label{algo:GrobBasis}
  \KwIn{$A$, $\theta$, $\delta$, a subalgebra $B\subset A$ over which $A=B[a_1,\dots,a_s]$ is free and over which $\theta,\delta$  are polynomial maps and the Gr\"obner basis can be computed, code parameters $n,k$.}
  \KwOut{A set of solutions $\cP=\{g,\hbar, f\ |\ \cC=Rg/Rf \text{ is dual-containing}\}$}
  $P_1\gets B[g_{0,1},\ldots, g_{0,s},\dots,g_{n-k-1,1},\ldots,g_{n-k-1,s},\hbar_{0,1},\ldots,\hbar_{0,s},\ldots, \hbar_{k-1,1},\ldots, \hbar_{k-1,s}]$    \Comment*[r]{multivariate ring over $B$}
  $\cP\gets \{\}$
  \tcc*[r]{Initialize a set to collect $g,\hbar$ of self-dual codes}
  \ForEach{$\hbar_k=\sum_{j=1}^s\hbar_{k,j}a_j\in \{$invertible element of $A$\} \label{line:LCchoice}}{
    $\algoVar{LSEs}\gets \{$Constraints such that $g_{i,j}, \hbar_{i,j}\in B$\} \Comment*[r]{$g_{i,j}^p=g_{i,j}, \hbar_{i,j}^p=\hbar_{i,j}$ if $B=\mathbb{F}_p$}
    $g\gets \sum_{i=0}^{n-k-1}(\sum_{j=1}^sg_{i,j}a_j)X^i+ X^{n-k}$
    \Comment*[r]{$g \in P_1[X;\theta,\delta]$}
    $\hbar\gets\sum_{i=0}^{k-1} (\sum_{j=1}^s\hbar_{i,j}a_j) X^i +(\sum_{j=1}^s\hbar_{k,j}a_j) X^{k}$\Comment*[r]{$\hbar \in P_1[X;\theta,\delta]$}
    $f\gets g\cdot \hbar$\Comment*[r]{$LC(f)$ may not be monic but does not contain variable}
    $h, r\gets$ quotient, remainder of $g$ right dividing $f$
    \Comment*[r]{$h,r\in P_1[X;\theta,\delta]$}
    $\algoVar{LSEs}\overset{\text{Append}}{\longleftarrow}$\{All coefficients of $r$ are $0$ \}
    \Comment*[r]{implies $g\mid_r f$\label{line:g_rdiv_f}}
    $\bG\gets$ a generator matrix constructed from $g$\;
    $\bM\gets$ the matrix constructed from $\hbar$ according to \cref{example_parity}\;
    $\algoVar{LSEs}\overset{\text{Append}}{\longleftarrow}$\{All entries in $\bM^\top\cdot \bM$ are $0$\}
    \tcc*[r]{implies $C^\perp\subseteq C$}
    $\cS\gets$ \{solutions of $g_{0,1},\ldots,g_{n-k-1,s},\hbar_{0,1},\ldots, \hbar_{k-1,s}$ from the Groebner basis of $\algoVar{LSEs}$\}\label{line:solutions}\;
    $\cP\overset{\text{Append}}{\longleftarrow}\{g,\hbar,f\in A[X;\theta,\delta]: \forall$ solution in $\cS\}$\;
    \tcc{$g,\hbar\in A[X;\theta,\delta]$ are reconstructed by evaluating coefficients of $g,\hbar\in P_1[X;\theta,\delta]$ for each solution in $\cS$; $f$ is reconstructed by $f=g\cdot \hbar$}
  }
\end{algorithm}
\subsubsection{Is the Dual $C^{\perp_{\sigma}}$ of a Cyclic Module $(\theta,\delta)$-Code Again a Cyclic Module $(\theta,\delta)$-Code ?}
\label{sec:dual_principle}
\label{algo1}
In the following let $\sigma$ be the identity or an automorphism of order $2$ of $A$. Note that the rows of the generator matrix $G^{\perp_{\sigma}}=\sigma(H)$ in \cref{th2} correspond to skew polynomials $p_1,\ldots,p_{k}$ in $R/Rf$ which form an $A$-basis of the free code $C^{\perp_{\sigma}}$. If the (Hermitian-) dual code $C^{\perp_{\sigma}}$ is a cyclic module code $Rg^{\perp_{\sigma}}/R\tilde{f}\subset R/R\tilde{f}$ generated by some monic skew polynomial $g^{\perp_{\sigma}}$ of degree $k$, then this monic polynomial ${g^{\perp_{\sigma}}}$ is a left divisor of all the polynomials $p_1,\ldots,p_{k}$.

We follow the notations in Algorithm \ref{algo:GrobBasis}. We first set up a monic polynomial $${g^{\perp_{\sigma}}}=\sum_{i=0}^{k-1}(\sum_{j=1}^sg^{\perp_{\sigma}}_{i,j}a_j)X^i+ X^{k}\in R$$ in the unknowns $g^{\perp_{\sigma}}_{0,1}, \dots, g^{\perp_{\sigma}}_{0,s}, \dots, g^{\perp_{\sigma}}_{k-1,s}$ over $B$. Then we perform a right division of all polynomials $p_1,\ldots,p_{n-k}$ by $g^{\perp_{\sigma}}$ and set the remainders $r_1,\ldots,r_{n-k}$ to zero. Note that the coefficients of all $r_1,\ldots,r_{n-k}$ are polynomials in $B[g^{\perp_{\sigma}}_{0,1}, \dots, g^{\perp_{\sigma}}_{k-1,s}]$ and must all be zero. This leads to a polynomial system over $B$ that can be solved by an algorithm via Gr\"obner basis. If the Gr\"obner basis is $\{1\}$ then $C^{\perp_{\sigma}}$ is not a cyclic module code in $R/Rf$ for any $f\in R$. Otherwise the Gr\"obner basis gives the generator polynomial $g^{\perp_{\sigma}}$ of degree $k$ of the cyclic module code $C^{\perp_{\sigma}}$.
\section{Computational Results for $A=\FF_2[v]/(v^2+v)$}
\label{results}
 We  keep the notation used in \cref{ex_F2v} and compute the dual-containing cyclic left module $(\theta,\delta)$-code over the ring $A=\FF_2[v]/(v^2+v)$  using the algorithm given in \cref{algo2}. Lemma \ref{primering}  shows that we can use the subalgebra $B=\FF_2\subset A$ for the algorithm.  Codes of small length over $A$ are classified in \cite{HU}. We follow \cite{DS} and define the Lee weight of $0,1,v,v+1$  respectively as $0,2,1,1$ and the Bachoc weight respectively as $0,1,2,2$.
\cref{tab:F2v2+v} and \cref{tab:F2v2+v_herm} give an overview of the best Euclidean and Hermitian dual-containing codes ${\cC}=Rg/Rf\subset R/Rf$ (the algorithm found all such codes). The empty set indicates that the approach shows that no dual-containing code $Rg/Rf\subset R/Rf$ exists for the parameter $[n,k]$. A question mark indicates that such dual-containing codes exist, but we did not compute the minimal distance.
The codes that could not have been found without considering non-zero derivations are marked in gray; the codes that only can be found by a non-zero derivation and an endomorphism which is not an automorphism are marked in dark gray.
Examples of such codes are the three weight distributions of the $[6,4]$ codes in \cref{tab:HammingF2v2}.
 Table \ref{tab:HammingF2v2} and \ref{tab:HammingF2v2_Herm} show more precisely which Hamming weight enumerators could only be found by using specific $(\theta,\delta)$ combinations.
    \begin{table}[htbp]
    \caption{Results on dual-containing cyclic module $(\theta,\delta)$-code over $\FF_2[v]/(v^2+v)$.}
     {\tiny
    \begin{subtable}[htbp]{\textwidth}
    \centering
       \subcaption{Best Hamming, Lee and Bachoc $d_H,d_L,d_B$ distance of dual-containing $(\theta,\delta)$-codes over $\mathbb{F}_2[v]/[v^2+v]$.
     \label{tab:F2v2+v}}
  $      \begin{array}{|l|c|c|c|c|c|c|c|c|c|c|c|}
     \hline
    n \setminus k & 2 & 3 & 4 &5 &6 &7& 8 & 9 & 10& 11& 12\\   \hline  \hline
         3 & 1,1,2 &\multicolumn{10}{c|}{}\\ \hline
      4 & 2,2,4&  2,2,2 &\multicolumn{9}{c|}{}\\  \hline
  5 &  & \emptyset  &\emptyset &\multicolumn{8}{c|}{} \\  \hline
  6 &  &  2,2,2  & 2,2,2   & 2,2,2  &\multicolumn{7}{c|}{}\\  \hline
  7 &  &  &  3,3,5  &\emptyset   & \emptyset &\multicolumn{6}{c|}{} \\  \hline
  8 &  &  &   4,4,7   & 2,2,4  & 2,2,2  &2,2,2  &\multicolumn{5}{c|}{} \\  \hline
  9 &  &  &    &   \emptyset &   \emptyset  &   \emptyset  &\cellcolor{gray!35}1,1,2 &\multicolumn{4}{c|}{}\\  \hline
  10 &  &  &    & 2,2,2   &  2,2,2   &   \emptyset    & \emptyset &2,2,2 &\multicolumn{3}{c|}{}\\  \hline
  11 &  &  &    &    &    \emptyset  &  \emptyset   &\emptyset  &\emptyset  &\emptyset  &\multicolumn{2}{c|}{} \\  \hline
  12 &  &  &    &   &   4,4,6   & 3,3,4 & 2,2,? &2,?,? &?,?,? & ?,?,?& \\  \hline
  13 &  &  &    &   &    &   \emptyset    &\emptyset &\emptyset & \emptyset  &\emptyset &\emptyset \\  \hline
       \end{array}%    }
     $
\end{subtable}}
    \begin{subtable}[htbp]{\textwidth}
    \centering
     \label{tab:exist_cyclic_dual}
\subcaption{For the dual-containing codes $C$, is $C^\bot$ a cyclic module code, according to \cref{sec:dual_principle}?}
       \begin{tabular}{|l|c|c|c|c|c|c|c|c|}
     \hline
    $n \setminus k$ & 2 & 3 & 4 &5 &6 &7& 8 & 9 \\   \hline  \hline
         3 & None &\multicolumn{7}{c|}{}\\ \hline
      4 & All & Some &\multicolumn{6}{c|}{}\\  \hline
  5 &  & /  & / &\multicolumn{5}{c|}{}\\  \hline
  6 &  &  All & Some & Some &\multicolumn{4}{c|}{} \\  \hline
  7 &  &  &  All & /   & / &\multicolumn{3}{c|}{}\\  \hline
  8 &  &  &  All  & Some  & Some & Some  &\multicolumn{2}{c|}{}\\  \hline
  9 &  &  &    &   /   & /   &  /   & None &\\  \hline
  10 &  &  &    & All  & Some  &   /     & /  & All \\  \hline
       \end{tabular}
\end{subtable}
{small
    \begin{subtable}[htbp]{\textwidth}
    \centering
    \label{tab:num_cyclic_dual}
       \subcaption{The number  of dual-containing $(\theta,\delta)$-codes and  codes whose dual is also a cyclic module $\theta,\delta)$ code.  (We only listed for the parameters marked with ``Some'' in \cref{tab:exist_cyclic_dual}(b) above.)}
       \begin{tabular}{|l|c|c|c|c|c|c|c|c|c|}
         \hline
         \multirow{2}{*}{$[n,k]$} &  \multicolumn{9}{c|}{\# of  dual-containing cyclic module  codes $Rg/Rf$ for each $(\theta,\delta)$} \\
                                  &  \multicolumn{9}{c|}{\# of above codes for which the dual code  is also a cyclic module $(\theta,\delta)$-code}   \\
                                    \hline &$(\mbox{Id};0)$ & $(\theta_2,0)$&$(\theta_2,\delta_2)$&$(\theta_2,\delta_3)$&$(\theta_2,\delta_4)$&$  (\theta_3,0)$&$(\theta_3,\delta_3)$&$(\theta_4,0)$&$(\theta_4,\delta_4)$ \\
         \hline
         \multirow{2}{*}{$[4,3]$} &  1 & 1&  3& 1& 1& 1& 2& 1& 2 \\
                                  &1& 1& 1& 1& 1& 1& 1& 1& 1 \\
         \hline
         \multirow{2}{*}{$[6,4]$} &  1& 1& 1& 2& 2& 1& 4& 1& 4  \\
         &  1& 1& 1& 1& 1& 1& 2& 1& 2  \\
         \hline
         \multirow{2}{*}{$[6,5]$} &  1& 1& 1& 2& 2& 1& 1& 1& 1  \\
         &  1& 1& 1& 1& 1& 1& 1& 1& 1  \\
         \hline
         \multirow{2}{*}{$[8,5]$} &  1& 3& 5& 1& 1& 1& 8& 1& 8  \\
         & 1& 3& 1& 1& 1& 1& 1& 1& 1  \\
         \hline
         \multirow{2}{*}{$[8,6]$} &  1& 3& 5& 1& 1& 1& 4& 1& 4  \\
         &  1& 3& 3& 1& 1& 1& 2& 1& 2  \\
         \hline
         \multirow{2}{*}{$[8,7]$} &  1& 1& 3& 1& 1& 1& 2& 1& 2  \\
         &  1& 1& 1& 1& 1& 1& 1& 1& 1  \\
         \hline
         \multirow{2}{*}{$[10,6]$} &  1& 1& 1& 1& 1& 1& 16& 1& 16 \\
                                  &  1& 1& 1& 1& 1& 1& 2& 1& 2 \\
         \hline
         % \multirow{2}{*}{$[10,9]$} & [ 1, 1, 1, 1, 1, 1, 1, 1, 1 ] \\
         %                          & [ 1, 1, 1, 1, 1, 1, 1, 1, 1 ] \\
         % \hline
       \end{tabular}
\end{subtable}}
\end{table}
  \begin{table}[htbp]
    \centering
    \caption{Best Hamming, Lee and Bachoc distance of $\theta_2$-Hermitian dual-containing $(\theta,\delta)$-codes over $\mathbb{F}_2[v]/[v^2+v]$
     \label{tab:F2v2+v_herm}}
     $
       \begin{array}{|l|c|c|c|c|c|c|c|c|}
     \hline
    n \setminus k & 2 & 3 & 4 &5 &6 &7& 8 & 9 \\   \hline  \hline
      4 & 2,2,4&  2,2,2&\multicolumn{6}{c|}{} \\  \hline
  5 &  & \cellcolor{gray!20}{2,2,2} &\cellcolor{gray!20}1,1,2&\multicolumn{5}{c|}{}\\  \hline
  6 &  &  \cellcolor{gray!35}3,3,4  & 2,2,4   & 2,2,2 &\multicolumn{4}{c|}{} \\  \hline
  7 &  &  &  3,3,5  &\cellcolor{gray!20}{1,1,2}   &\cellcolor{gray!20} {1,1,2} &\multicolumn{3}{c|}{}\\  \hline
  8 &  &  &   3,3,6   & 2,2,4  & 2,2,2  &2,2,2 &\multicolumn{2}{c|}{} \\  \hline
  9 &  &  &    &   \cellcolor{gray!35}{1,1,2}&   \emptyset  &   \emptyset  &\emptyset &\multicolumn{1}{c|}{}\\  \hline
  10 &  &  &    & 2,2,2   &  2,2,2   &   \emptyset    & \emptyset &2,2,2 \\  \hline
       \end{array}
     $
\end{table}
\begin{table}[htbp]
    \centering
  \caption{Hamming weight enumerator of dual-containing $(\theta,\delta)$-codes over $\mathbb{F}_2[v]/[v^2+v]$.
    \label{tab:HammingF2v2}}
  $
  \begin{array}{|c|l|l|} \hline
[n,k]  & \mbox{Hamming Weight}    & \mbox{Constructed with }  (\theta,\delta)\\ \hline \hline
    \multirow{2}{*}{[4,2]} & 1+6w^2+9w^4
                                  & \text{all combinations $(\theta,\delta)$ provide such an example}  \\ \cline{2-3}
& \cellcolor{gray!20} 1+4w^2+4w^3+7w^4
                                  & (\theta_2,\delta_2), (\theta_3,\delta_3),(\theta_4,\delta_4)   \\ \hline\hline
    \multirow{1}{*}{[6,3]} & 1+9w^2+27w^4+\ldots %+ 27*w^6
                                  & \text{all combinations  $(\theta,\delta)$ provide such an example} \\ \hline\hline
    \multirow{4}{*}{[6,4]} & 1+9w^2+24w^3+\ldots %\item [ 1, 0, 9, 24, 99, 72, 51 ]
                                  & \text{all combinations  $(\theta,\delta)$ provide such an example} \\ \cline{2-3}
& \cellcolor{gray!20} 1+17w^2+24w^3+\ldots  %\item [ 1, 0, 17, 24, 83, 72, 59 ]
                                  & (\theta_2,\delta_3),(\theta_2,\delta_3)   \\ \cline{2-3}
& \cellcolor{gray!35} 1+2w+11w^2+\ldots  %\item [ 1, 2, 11, 28, 87, 66, 61 ]
                                  & (\theta_3,\delta_3),(\theta_4,\delta_4)    \\ \cline{2-3}
&\cellcolor{gray!35}  1+13w^2+24w^3+\ldots  %\item [ 1, 0, 13, 24, 91, 72, 55 ]
                                  & (\theta_3,\delta_3),(\theta_4,\delta_4)  \\ \hline \hline
    \multirow{4}{*}{[8,4]} & 1+12w^2+54w^4+\ldots  %\item [ 1, 0, 12, 0, 54, 0, 108, 0, 81 ]
                                  & \text{all combinations  $(\theta,\delta)$ provide such an example} \\ \cline{2-3}
& 1+28w^4+56w^5+\ldots  %\item [ 1, 0, 0, 0, 28, 56, 84, 56, 31 ]
                                  &  (\theta_2,0)   \\ \cline{2-3}
& \cellcolor{gray!20} 1+4w^2+38w^4+\ldots  %\item [ 1, 0, 4, 0, 38, 32, 100, 32, 49 ]
                                  & (\theta_2,\delta_2),(\theta_3,\delta_3),(\theta_4,\delta_4)   \\ \hline
 \end{array}$
\end{table}

\begin{table}[htbp]
    \centering
  \caption{Hamming weight enumerator of $\theta_2$-Hermitian dual-containing $(\theta,\delta)$-codes over \\$\mathbb{F}_2[v]/[v^2+v]$.
    \label{tab:HammingF2v2_Herm}}
  $
\begin{array}{|c|l|l|} \hline
[n,k]  & \mbox{Hamming Weight}    & \mbox{Constructed with }  (\theta,\delta)\\ \hline \hline
\multirow{1}{*}{[4,2]} & 1+6w^2+9w^4
& \text{all combinations  $(\theta,\delta)$ provide such an example} \\
\cline{2-3}&1+2w^2+8w^3+5w^4 &
(\theta_{2},0)\\
\hline
\multirow{1}{*}{[4,3]} & 1+18w^2+\ldots
& \text{all combinations  $(\theta,\delta)$ provide such an example} \\
\cline{2-3}&\cellcolor{gray!20}1+2w+16w^2+\ldots&
(\theta_{2},\delta_2),(\theta_3,\delta_3),(\theta_4,\delta_4)\\
\cline{2-3}&\cellcolor{gray!20}1+2w+12w^2+\ldots &
(\theta_{2},\delta_3), (\theta_2,\delta_4)\\
\hline
\multirow{1}{*}{[5,3]} & \cellcolor{gray!20}1+8w^2+14w^3+\ldots
& (\theta_2,\delta_3),(\theta_2,\delta_4) \\
\cline{2-3}&\cellcolor{gray!35}1+w+6w^2+\ldots &
(\theta_{3},\delta_{3}),
(\theta_{4},\delta_{4})\\
\hline
\multirow{1}{*}{[5,4]} &\cellcolor{gray!20}1+3w+22w^2+\ldots &
(\theta_{2},\delta_{3}),
(\theta_{2},\delta_{4})\\
\hline
\multirow{1}{*}{[6,3]} & 1+9w^2+27w^4+\ldots
& \text{all combinations  $(\theta,\delta)$ provide such an example} \\
\cline{2-3}&\cellcolor{gray!35}1+8w^3+21w^4+\ldots &
(\theta_{3},\delta_{3}),
(\theta_{4},\delta_{4})\\
\hline
 \end{array}
 $
\end{table}

    \begin{table}[htbp]
  \centering
  \caption{
$[6,4]$ dual-containing $(\theta,\delta)$-codes over $\mathbb{F}_2[v]/[v^2+v]$ with Hamming weight enumerator $1+13w^2+24w^3+\dots$.}
  \label{tab:gG_n6_k4}
  \begin{tabular}{l|c|c|c}
Index & $g$& $(\theta,\delta)$ &$\mathbf{G}$  \\
\hline
1& $g=X^2 + X + v + 1$ &  $(\theta_{3},\delta_{3})$
  & $\begin{pmatrix}
v + 1 & 1 & 1 & 0 & 0 & 0\\
v & 1 & 1 & 1 & 0 & 0\\
v & 0 & 1 & 1 & 1 & 0\\
v & 0 & 0 & 1 & 1 & 1
\end{pmatrix}
                    $ \\
    \hline
2& $g=X^2 + (v + 1)X + 1$   & $(\theta_{3},\delta_{3})$
    & $\begin{pmatrix}
1 & v + 1 & 1 & 0 & 0 & 0\\
0 & v + 1 & 1 & 1 & 0 & 0\\
0 & v & 1 & 1 & 1 & 0\\
0 & v & 0 & 1 & 1 & 1
\end{pmatrix}
$ \\
    \hline
3&    $g=X^2 + vX + 1$
&$(\theta_{4},\delta_{4})$
&$\begin{pmatrix}
1 & v & 1 & 0 & 0 & 0\\
0 & v & 1 & 1 & 0 & 0\\
0 & v + 1 & 1 & 1 & 1 & 0\\
0 & v + 1 & 0 & 1 & 1 & 1
\end{pmatrix}
$ \\
    \hline
4& $g=X^2 + X + v$   &
$(\theta_{4},\delta_{4})$
 &$\begin{pmatrix}
v & 1 & 1 & 0 & 0 & 0\\
v + 1 & 1 & 1 & 1 & 0 & 0\\
v + 1 & 0 & 1 & 1 & 1 & 0\\
v + 1 & 0 & 0 & 1 & 1 & 1
\end{pmatrix}
$
  \end{tabular}
\end{table}

Besides the complexity of a  Gr\"obner basis,  the complexity of our approach is  also  linked to the fact that many decompositions $f=hg=g \hbar$ can exist for a fixed $g$, and that all combinations lead to the same code whose generating matrix is constructed only from $g$ and the corresponding $A[X;\theta,\delta]$. To illustrate this we present in more detail the results for the $[6,4]$ code with Hamming weight enumerator $1+13w^2+24w^3+\dots$\ . There are four possible generator polynomials $g$ presented in \cref{tab:gG_n6_k4}. Note that given $g$ and $A[X;\theta,\delta]$ one can compute a generating matrix immediately. We consider the first polynomial $g=X^2 + X + v + 1\in \left(\mathbb{F}_2[v]/[v^2+v]\right)[X; \theta_{3},\delta_{3}]$ of  \cref{tab:gG_n6_k4}. There exist $8$ \emph{non-central} polynomials $f$ for which there are polynomials $h,\hbar$ such that $f=hg=g\hbar$, i.e., $g$ is a left and right divisor of $f$:
\begin{eqnarray*}
  f_1& = & X^6 + vX^4 + vX^3 + vX + v + 1 = (X^4 + X^3 + vX^2 + X + v + 1) \cdot g\\
  f_2 & = & X^6 + X^5 + (v + 1) X^4 + X^3 + v X + v + 1 \\
  & = & (X^4 + v X^2 + (v + 1) X + 1) \cdot g\\
  f_3& = & X^6 + (v + 1) X^4 + v X^3 + v X^2 + X + v + 1\\
  & = & (X^4 + X^3 + (v + 1) X^2 + 1) \cdot g\\
   f_4 &=& X^6 + X^5 + v X^4 + X^3 + v X^2 + X + v + 1\\
   & = & (X^4 + (v + 1) X^2 + v X + v + 1 ) \cdot g\\
   f_5 &=& X^6 + v X^4 + v X^3 + X^2 + (v + 1) X = (X^4 + X^3 + v X^2 + X + v) \cdot g\\
   f_6 &= &X^6 + X^5 + (v + 1) X^4 + X^3 + X^2 + (v + 1) X \\ & = & (X^4 + v X^2 + (v + 1) X) \cdot g\\
   f_7 &= &X^6 + (v + 1) X^4 + v X^3 + (v + 1) X^2 = (X^4 + X^3 + (v + 1) X^2) \cdot g\\
    f_8& =& X^6 + X^5 + v X^4 + X^3 + (v + 1) X^2 = (X^4 + (v + 1) X^2 + v X + v) \cdot g
\end{eqnarray*}
 For each $f$, there is a unique $h$ corresponding to $f=hg$ and 16 distinct $\hbar$ such that $f=g\hbar$, where one of $\hbar$ is equal to $h$.
The following are the other 15 distinct $\hbar\neq h$ such that $f_1=g\cdot \hbar$:
 \begin{align*}
   f_1 %&= g\cdot X^4 + X^3 + vX^2 + X + v + 1\\
   &= g\cdot \left(X^4 + (v + 1) X^3 + v X^2 + (v + 1) X + v + 1\right)\\
   &= g\cdot\left( X^4 + X^3 + v X^2 + (v + 1) X + v + 1\right)\\
        &=g\cdot\left(X^4 + (v + 1) X^3 + (v + 1) X + v + 1 \right)
   = g\cdot \left(X^4 + X^3 + (v + 1) X + v + 1\right)\\
        &=g\cdot \left(X^4 + (v + 1) X^3 + v X^2 + X + v + 1\right)
   = g\cdot \left(X^4 + (v + 1) X^3 + X + v + 1 \right)\\
        &= g\cdot \left(X^4 + X^3 + X + v + 1 \right)
   = g\cdot \left( X^4 + (v + 1) X^3 + v X^2 + (v + 1) X + 1\right)\\
       & = g\cdot \left(X^4 + X^3 + v X^2 + (v + 1) X + 1 \right)
   = g\cdot \left( X^4 + (v + 1) X^3 + (v + 1) X + 1 \right)\\
       & = g\cdot \left(X^4 + X^3 + (v + 1) X + 1 \right)
   = g\cdot \left( X^4 + (v + 1) X^3 + v X^2 + X + 1\right)\\
       & = g\cdot \left(X^4 + X^3 + v X^2 + X + 1  \right)
   = g\cdot \left(X^4 + (v + 1) X^3 + X + 1 \right)\\
       & = g\cdot \left(  X^4 + X^3 + X + 1\right)
 \end{align*}

The  $[8,4,d_H=5]$ code in  \cref{tab:F2v2+v_herm} that achieves the Singleton bound in Hamming metric are obtained with
 $g=X^4 + (v + 1) X^3 + X^2 + v X + 1$ and $g=X^4 + v X^3 + X^2 + (v + 1) X + 1$ in $\left(\mathbb{F}_2[v]/[v^2+v]\right)[X; \mbox{id},\theta_{2}]$.

 We apply the algorithm presented in \cref{algo1} to verify for which  dual-containing codes $C=Rg/Rf\subset R/Rf$ the dual code $C^\bot$ is again a cyclic module $(\theta,\delta)$-code. See in \cref{tab:exist_cyclic_dual}b and \cref{tab:num_cyclic_dual}c for an overview. We list two examples below which show that the dual of a dual-containing cyclic module $(\theta,\delta)$-code is not always
 a cyclic module $(\theta,\delta)$-code:
\begin{itemize}
\item
   For $[n=4,k=3]$, we found three $g\in A[X; \theta_2,\delta_2]$ that generate dual-containing cyclic module codes:
$     g_1=X + v + 1$, $ g_2=X + 1$, $g_3=X + v$
   where only the dual of $g_2 = X + 1$ is a cyclic module code, with $g_2^{\perp}=X^3 + X^2 + X + 1$.
\item For $[n=6, k=4]$, we found four $g\in A[X;\theta_3,\delta_3]$ that generate dual-containing cyclic module codes:
$ g_1 =X^2 + (v + 1) X + v + 1$,
$g_2=X^2 + X + 1$,
$g_3=X^2 + X + v + 1$,
$g_4=X^2 + (v + 1)X + 1$.
Only the dual of $g_2$ and $g_4$ are cyclic module codes, with $g_2^{\perp}=X^4 + X^3 + X + 1$ and $g_4^{\perp}= X^4 + (v + 1)X^3 + X + v + 1$, respectively.
 \end{itemize}

\section{Computational Results for $A=\FF_2[u]/(u^2)$}
We keep the notations of \cref{eg:ord4_u2_0}
% \cref{eg:ord4_u2}
for the ring $A=\FF_2[u]/(u^2)$.
Lemma \ref{primering} and Algorithm \ref{algo:GrobBasis} show that we can search dual-containing cyclic left module $(\theta,\delta)$-code ${\cC}=Rg/Rf\subset R/Rf$ over the subalgebra $B=\FF_2\subset A$ using a Gr\"obner basis approach. We follow \cite{DS} and define the Lee weight of $0,1,u,u+1$  respectively as $0,1,2,1$ and the Euclidean weight respectively as $0,1,4,1$.
\cref{tab:F2u2} give an overview of the best dual-containing codes ${\cC}=Rg/Rf\subset R/Rf$ (the algorithm found all such codes).
For the cell marked in gray, the dual-containing cyclic module codes are only found from the maps $(\mbox{id},\delta_2)$ and $(\mbox{id},\delta_4)$.
\begin{table}[htbp]
  \caption{Best Hamming, Lee, and Euclidean distances of dual-containing cyclic module $(\theta,\delta)$-codes over $\FF_2[u]/(u^2)$.
     \label{tab:F2u2}}
     $$
       \begin{array}{|l|c|c|c|c|c|c|c|c|}
     \hline
    n \setminus k & 2 & 3 & 4 &5 &6 &7& 8 & 9 \\   \hline  \hline
      4 &  2,4,4& 2,2,2&\multicolumn{6}{c|}{} \\  \hline
  5 &  & \emptyset  &\cellcolor{gray!20}1,2,2 &\multicolumn{5}{c|}{}\\  \hline
  6 &  &  2,4,4 & 2,2,2 & 2,2,2 &\multicolumn{4}{c|}{} \\  \hline
         7 &  &  &  3,3,3 & \emptyset   & {1,2,2}&\multicolumn{3}{c|}{}
                                          % \textcolor{red}{\theta_1,\delta_2 }
         \\  \hline
  8 &  &  &  4,4,4  & 2,4,4  & 2,2,2 &2,2,2 &\multicolumn{2}{c|}{}  \\  \hline
  9 &  &  &    &   \emptyset   & \emptyset   &  \emptyset   &\cellcolor{gray!20}1,2,2 &\\  \hline
  10 &  &  &    & 2,4,6  & 2,4,5  &   \emptyset     & \emptyset  &2,2,2 \\  \hline
       \end{array}
     $$
    \end{table}

 \cref{tab:exist_cyclic_dual_F2u2} gives an overview whether the dual codes $C^\bot$ of the dual-containing codes found by Algorithm \ref{algo:GrobBasis} are again cyclic module $(\theta,\delta)$-codes.
 \begin{table}[htbp]
\caption{For which dual-containing cyclic module codes $C$ is the dual $C^\bot$ again a cyclic module code, according to \cref{sec:dual_principle}?}
\label{tab:exist_cyclic_dual_F2u2}
\begin{center}
       \begin{tabular}{|l|c|c|c|c|c|c|c|c|}
     \hline
    $n \setminus k$ & 2 & 3 & 4 &5 &6 &7& 8 & 9 \\   \hline  \hline
      4 & All & All &\multicolumn{6}{c|}{} \\  \hline
  5 &  & /  & {
              \begin{tabular}{@{}c@{}}
                $(\mbox{id},\delta_2)$: None\\
                $(\mbox{id},\delta_4)$: All
                %\\ other $(\theta,\delta)$: no codes found
              \end{tabular}}&\multicolumn{5}{c|}{}\\
                 \hline
  6 &  &  All  & All  & All &\multicolumn{4}{c|}{}  \\  \hline
  7 &  &  &  All  & /   & All &\multicolumn{3}{c|}{} \\  \hline
  8 &  &  &  All   & All   & All  &All &\multicolumn{2}{c|}{}   \\  \hline
         9 &  &  &    &   /   & /   &  /   & {
                                             \begin{tabular}{@{}c@{}}
                                              $(\mbox{id},\delta_2)$: None \\
                                              $(\mbox{id},\delta_4)$: Some
                                            % \\ other $(\theta,\delta)$: no codes found
                                             \end{tabular}} &\\  \hline
  10 &  &  &    & All   & All   &   /     & /  &All  \\  \hline
       \end{tabular}\end{center}
\end{table}

 Table \ref{tab:HammingF2u2} shows examples for which Hamming weight enumerators could only be found using specific $(\theta,\delta)$ combinations. In particular the gray cells indicate those that can only be obtained using a non-zero derivation.
\begin{table}[htbp]
  \caption{Hamming weight enumerator of dual-containing $(\theta,\delta)$-codes over $\mathbb{F}_2[u]/[u^2]$.
    \label{tab:HammingF2u2}}
  $$
  \begin{array}{|c|l|l|} \hline
[n,k]  & \mbox{Hamming Weight}    & \mbox{Constructed with }  (\theta,\delta)\\ \hline \hline
\multirow{2}{*}{[4,2]} & 1+2w^2+8w^3+5w^4  & (\Id,0), (\Id,\delta_2), (\Id,\delta_3), (\theta_2,\delta_2)   \\ \cline{2-3}
       &   1+6w^2+9w^4  & \text{all maps} %(\Id,0),  (\Id,\delta_2), (\Id,\delta_3), (\Id,\delta_4), (\theta_2,0), (\theta_2,\delta_2)
    \\ \hline  \hline
    \multirow{5}{*}{[8,4]} & 1+ 4w^2+30w^4+\ldots %64w^5+52w^6+ 64w^7+ 41w^8 &
                       &  (\Id,0), (\theta_2,\delta_2)   \\ \cline{2-3}
       & 1+ 4w^2+46w^4+ \ldots %148+w^6+ 57w^8 &
                       & (\Id,0)   \\ \cline{2-3}
       & 1+ 4w^2+ 16w^3+ \ldots %+ 14w^432w^5+ 84w^6+ 80w^7+ 25w^8 &
                       &   (\Id,0)   \\ \cline{2-3}
       & 1+ 12w^2+54w^4+ \ldots%108w^6+ 81w^8 &
                       &  \text{all maps}  \\ \cline{2-3}
       &\cellcolor{gray!20}  1+ 26w^4+64w^5+\ldots %  72w^6+ 64w^7+ 29w^8 &
                       &  (\Id,\delta_2)   \\ \hline  \hline
    \multirow{4}{*}{[8,5]} & 1+4w^2+16w^3+94w^4+\ldots %\item [ 1, 0, 4, 16, 94, 224, 308, 272, 105 ]
                                  &  (\Id,0),(\Id,\delta_2)   \\ \cline{2-3}
       & 1+4w^2+16w^3+110w^4+\ldots %\item [ 1, 0, 4, 16, 110, 160, 404, 208, 121 ]
                                     &     (\Id,0)   \\ \cline{2-3}
       & 1+12w^2+102w^4+\ldots  %\item [ 1, 0, 12, 0, 102, 192, 396, 192, 129 ]
                                  &   \text{all maps}  \\ \cline{2-3}
       &\cellcolor{gray!20} 1+16w^2+8w^3+114w^4+\ldots %\item [ 1, 0, 16, 8, 114, 176, 360, 200, 149 ]
                            & { (\Id,\delta_2) } \\ \hline
  \end{array}
  $$
\end{table}

\section{Computational Results for $A=\FF_4$}
 Consider the field $\FF_4=\FF_2(\alpha)$  where $\alpha^2=\alpha+1$. There are two automorphisms: $\theta_1=\Id$ and the Frobenius automorphism $\theta_2: x\mapsto x^2$. $\theta_2$ is of order $2$ and is a polynomial map on both $\FF_4$ and $\FF_2$. All the $\theta$-derivations are inner derivations and are marked in gray in the list below:
  \begin{equation*} %\label{eq:F4maps}
  \begin{array}{|c||l|l|l|l|} \hline
      &\multicolumn{2}{c|}{Automorphism} \\ \hline
 & \theta_1=\Id & \theta_2(\alpha)=\alpha+1 \\ \hline \hline
\delta_1=0&   \cellcolor{gray!30}{\alpha\mapsto 0} &  \cellcolor{gray!30}{\alpha\mapsto 0
%\ (\beta=0)
}  \\ \hline
\delta_2 & &  \cellcolor{gray!30}{\alpha\mapsto 1
%\ (\beta = 1)
}   \\ \hline
\delta_3 & &  \cellcolor{gray!30}{\alpha\mapsto \alpha
%\ (\beta = \alpha)
}  \\ \hline
\delta_4 & &  \cellcolor{gray!30}{\alpha\mapsto \alpha+1
%\ (\beta = \alpha+1)
} \\ \hline \end{array}
  \end{equation*}

  Following \cite{DS} we define the Lee weight of $0,1,\alpha,\alpha+1$ respectively as $0,2,1,1$ and following \cite{LingSole2001} we define the Euclidean weight respectively as $0,1,2,1$.

  \cref{tab:F4_herm} shows the existence and the best Hamming, Lee and Euclidean distance of the $\theta_2$-Hermitian dual-containing cyclic module $(\theta,\delta)$-codes $\cC=Rg/Rf\subset R/Rf$ over $\FF_4$.
  The gray cells indicate the codes that can only be obtained using a non-zero derivation.
  \begin{table}[htbp]
    \caption{The best Hamming, Lee and Euclidean $d_H,d_L,d_E$ distance of $\theta_2$-Hermitian dual-containing codes $Rg/Rf\subset R/Rf$ over $\mathbb{F}_4$.
     \label{tab:F4_herm}}
     $$
       \begin{array}{|l|c|c|c|c|c|c|c|c|}
     \hline
    n \setminus k & 2 & 3 & 4 &5 &6 &7& 8 & 9\\   \hline  \hline
      4 & 2,2,2&  2,2,2 &\multicolumn{6}{c|}{}\\  \hline
  5 &  & 3,3,3  &\cellcolor{gray!20}1,1,1 &\multicolumn{5}{c|}{} \\  \hline
  6 &  &  4,4,4  & 2,2,2   & 2,2,2  &\multicolumn{4}{c|}{}\\  \hline
  7 &  &  &  3,3,3  & \emptyset  & \cellcolor{gray!20}1,1,1 &\multicolumn{3}{c|}{} \\  \hline
  8 &  &  &   2,2,2  & 2,2,2  & 2,2,2 &2,2,2 &\multicolumn{2}{c|}{}\\  \hline
  9 &  &  &    &   \emptyset&   \emptyset &   \emptyset  &\cellcolor{gray!20}1,1,1& \\  \hline
  10 &  &  &    &  4,4,4  &  3,3,3 & 2,2,2     & 2,2,2 & 2,2,2\\  \hline
       \end{array}
     $$
\end{table}

 \cref{tab:HammingF4} provides some examples of the Hamming weight distributions.
\begin{table}[htbp]
  \caption{Weight enumerator of $\theta_2$-Hermitian dual-containing cyclic module $(\theta,\delta)$ codes  over $\mathbb{F}_4$.
    \label{tab:HammingF4}}
  $$
  \begin{array}{|c|l|c|} \hline
[n,k]  & \mbox{Hamming Weight Enumerator}    & \mbox{Constructed with }  (\theta,\delta)\\ \hline \hline
    \multirow{2}{*}{[4,3]}
       & 1+ 18w^2 +24w^3 +211w^4
                                             & \text{all maps}   \\ \cline{2-3}
 &\cellcolor{gray!20} 1+ 6w+ 12w^2+ 18w^3+ 27w^4
                                             & (\theta_2,\delta_2)   \\ \hline\hline
    \multirow{1}{*}{[5,4]}
       & \cellcolor{gray!20}1+ 9w+ 30w^2+ 54w^3+ 81w^4+81w^5
    & (\theta_2,\delta_2 ) \\
    \hline\hline
    \multirow{2}{*}{[6,5]}
    & 1+ 45w^2 +120w^3+ 315w^4 +360w^5 +183w^6 & \text{all maps}   \\ \cline{2-3}
      & \cellcolor{gray!20}1+ 12w+ 57w^2+ 144w^3+ 243w^4+\ldots
    & (\theta_2,\delta_2 ) \\ \hline\hline
    \multirow{1}{*}{[7,6]}
&\cellcolor{gray!20} 1+ 15w+ 93w^2+ 315w^3+ 675w^4 + \ldots &  (\theta_2,\delta_2 )\\
     \hline     \hline
     \multirow{2}{*}{[8,7]}
     & 1+ 84w^2+ 336w^3+ 1470w^4+ \ldots &\text {all maps}\\ \cline{2-3}
&\cellcolor{gray!20} 1+ 18w+ 138w^2+ 594w^3+ 1620w^4+ \ldots
&(\theta_2,\delta_2 )\\
     \hline\hline
     \multirow{1}{*}{[9,8]}
& \cellcolor{gray!20} 1+ 21w+ 192^2+ 1008w^3+ 3402w^4+\ldots &  (\theta_2,\delta_2 )\\
     \hline     \hline
      \multirow{2}{*}{[10,9]}
     & 1+ 135w^2+ 720w^3+ 4410w^4+ 15120w^5+\ldots &\text {all maps}\\ \cline{2-3}
& \cellcolor{gray!20} 1+ 24w+ 255w^2+ 1584w^3+ 6426w^4+ \ldots
&(\theta_2,\delta_2 )\\       \hline
 \end{array}
 $$
\end{table}
 This shows that in Hermitian case, non-zero derivation does produce other code than in the $\delta=0$ case. In the $\delta=0$ case we could not exhibit new codes.

\section{Computation Results for the Galois Ring $A=GR(4,2)$}
The galois ring $A=GR(4,2)=Z_4[u]=(\ZZ/4\ZZ)[u]/(u^2+u+1)$ is a Frobenius ring  of order $16$. This ring has two automorphisms:
$\theta_1=\Id$ and $\theta_2(u)=3u+3$ of order $2$.
% $\theta_2$ is isomorphic to the cyclic group $C_2$ of order $2$.
The zero derivation is the only $\mbox{id}$-derivation. The $\theta_2$-derivations are all inner (i.e.~$\delta: a\mapsto \beta a-\theta_2(a)\beta, \forall \beta\in A$):
$\delta_1(u) = {0}$, $\delta_2(u) = {u}$, $\delta_3(u) ={2u} $, $\delta_{4}(u) = {3u} $, $\delta_5(u) =   {1}$, $\delta_{6}(u) = {u+1}$, $\delta_{7}(u) =  {2u+1} $, $\delta_8(u) =  {3u+1}$, $\delta_9(u) =  {2} $, $\delta_{10}(u) = {u+2} $, $\delta_{11}(u) = {2u+2}$, $\delta_{12}(u) = {3u+2 } $, $\delta_{13}(u) ={ 3}$, $\delta_{14}(u) ={u+3} $, $\delta_{16}(u) = {3u+3 }$.

We computed all $[4,2]$ self-dual and $[4,3]$ dual-containing cyclic left module $(\theta,\delta)$-codes over $A=GR(4,2)$ by Algorithm \ref{algo:GrobBasis}.
For the $[4,2]$ codes, there are 8 $g$'s given by each map in \cref{tab:existanceZ4_x2_x_1}. They generate $[4,2,d_H=3]$ codes with distinct codebooks (i.e.~the codewords in the codes are not all equal), however, with the same weight enumerator $1+60w^3+195w^4$. For each $g$ there are 16 $f$'s which are all central and including one in the form of $X^n-a$ for some $a\in\ZZ_4$.
The $[4,3]$ codes can be obtained from all maps. From each map, there are four unique $g$'s. For each $g$ there are more than 1000 $f$'s which include at least one central $f$ and an $f$ in the form of $X^n-a$ for some $a\in\ZZ_4$. All the codes have the same weight enumerator given in \cref{tab:existanceZ4_x2_x_1}.

  \begin{table}[htbp]
  \caption{The best Hamming distance $d_H$ of dual-containing codes $Rg/Rf\subset R/Rf$ over $GR(4,2)$.
     \label{tab:existanceZ4_x2_x_1}}
     $$
  %   \textcolor{blue}{
       \begin{array}{|l|c|c|c|c|}
         \hline
         [n,k]  & \text{existing code for map }(\theta_i,\delta_j) & \text{best }d_H& \text{Weight Distribution}\\   \hline  \hline
  \multirow{2}{*}        {[3,2]} & (1,1), (2,2), (2,4), (2,6), (2,8), & \multirow{2}{*}{2}  & 1 + 45w^2 + 210w^3\\
  & (2,10), (2,12), (2,14), (2,16) & &  \\
         \hline
         [4,2] &  (2, 1), (2, 3), (2,9), (2,11) & 3  & 1 + 60w^3 + 195w^4\\
         \hline
         [4,3] & \text{All maps} & 2 &1 + 90w^2 + 840w^3 + 3165w^4\\
         \hline
         [5,3] & \emptyset & / & / \\
         \hline
       \end{array}
%     }
     $$
\end{table}

\bibliographystyle{IEEEtranS}
% \nocite{*} % list all references in .bib file
\bibliography{refs}
\end{document}

%% file: def.tex
\usepackage{url}

\usepackage{geometry}
\usepackage{amsmath, amsthm}
\usepackage{amsfonts, amssymb}
\usepackage[table]{xcolor}
\usepackage{array}
\usepackage{longtable}
\newlength{\tablcol}
\usepackage{tablefootnote}
\usepackage[ruled,lined, linesnumbered]{algorithm2e}
\usepackage{subcaption}
\SetKwComment{Comment}{/* }{ */}

\def\FF{{\mathbb F}} \def\ZZ{{\mathbb Z}} \def\NN{{\mathbb N}}
  
\newtheorem{theorem}{Theorem}
\newtheorem{proposition}{Proposition}

\newtheorem{definition}{Definition}

\newtheorem{lemma}{Lemma}

\newtheorem{example}{Example}

\usepackage[hidelinks]{hyperref}
\usepackage[capitalize]{cleveref}
\usepackage{color,xcolor,graphicx,fancybox,texdraw}
\usepackage{multirow}
\usepackage{makecell}
\usepackage[]{authblk}

\DeclareMathOperator{\lker}{lker}

\DeclareMathOperator{\Id}{id}

\newcommand{\coloneq}{:=}   
\newcommand{\ie}{i.e.~}

\newcommand\blfootnote[1]{%
  \begingroup
  \renewcommand\thefootnote{}\footnote{#1}%
  \addtocounter{footnote}{-1}%
  \endgroup
}

\newcommand{\cC}{\mathcal{C}}

\newcommand{\cE}{\mathcal{E}}

\newcommand{\cP}{\mathcal{P}}

\newcommand{\cS}{\mathcal{S}}

\newcommand{\bG}{\boldsymbol{G}}

\newcommand{\bM}{\boldsymbol{M}}

\newcommand{\bc}{\boldsymbol{c}}

\newcommand{\bv}{\boldsymbol{v}}
\newcommand{\bw}{\boldsymbol{w}}
\newcommand{\bx}{\boldsymbol{x}}
\newcommand{\by}{\boldsymbol{y}}

\newcommand{\0}{\boldsymbol{0}}

\newcommand{\algoVar}[1]{\mathsf{#1}} % algorithm variables

\newif\ifcomment

\definecolor{TUMBlau}{RGB}{0,101,189} % Pantone 300 (0, 0.396, 0.7412)
\definecolor{TUMBlauDunkel}{RGB}{0,82,147} % Pantone 301 (0, 0.3216,0.5765)
\definecolor{TUMBlauHell}{RGB}{152,198,234} % Pantone 283 (0.596,0.776,0.918)
\definecolor{TUMBlauMittel}{RGB}{100,160,200} % Pantone 542 (0.392, 0.627, 0.784)

\definecolor{TUMElfenbein}{RGB}{218,215,203} % Pantone 7527 -Elfenbein
\definecolor{TUMGruen}{RGB}{162,173,0} % Pantone 383 - Grün (0.6353,0.6784,0)
\definecolor{TUMOrange}{RGB}{227,114,34} % Pantone 158 - Orange (0.8902, 0.4471, 0.133)
\definecolor{TUMGrau}{gray}{0.6} % Grau 60%

\definecolor{TUMGruenDunkel}{RGB}{0,124,48} % (0,0.4863,0.1882)
\definecolor{TUMRot}{RGB}{196,7,27} % (0.7686,0.02745,0.10588)

\definecolor{TUMBlue}{RGB}{0,101,189} % Pantone 300 (0, 0.396, 0.7412)
\definecolor{TUMBlueDark}{RGB}{0,82,147} % Pantone 301 (0, 0.3216,0.5765)
\definecolor{TUMBlueLight}{RGB}{152,198,234} % Pantone 283 (0.596,0.776,0.918)
\definecolor{TUMBlueMedium}{RGB}{100,160,200} % Pantone 542 (0.392, 0.627, 0.784)

\definecolor{TUMIvory}{RGB}{218,215,203} % Pantone 7527 -Elfenbein
\definecolor{TUMGreen}{RGB}{162,173,0} % Pantone 383 - Grün (0.6353,0.6784,0)
\definecolor{TUMGray}{gray}{0.6} % Grau 60%
\definecolor{TUMGrayDark}{gray}{0.3} % Grau 80%

\definecolor{TUMGreenDark}{RGB}{0,124,48} % (0,0.4863,0.1882)
\definecolor{TUMRed}{RGB}{196,7,27} % (0.7686,0.02745,0.10588)